\pdfoutput=1

\newif\ifFull
\Fullfalse

\ifFull
\documentclass[12pt]{article}
\else
\documentclass{jgaa-art}
\fi

\usepackage{graphicx}
\usepackage{hyperref}
\usepackage{amsfonts, amssymb, amsmath}
\usepackage{cite}

\ifFull
\usepackage{fullpage}
\usepackage{amsthm}
\newtheorem{theorem}{Theorem}
\newtheorem{lemma}[theorem]{Lemma}
\theoremstyle{definitions}

\fi

\renewcommand{\P}{\mathsf{P}}
\newcommand{\NP}{\mathsf{NP}}
\usepackage{color}

\begin{document}

\title{Inapproximability of Orthogonal Compaction}

\ifFull
\author{Michael J. Bannister \and David Eppstein \and Joseph A. Simons \\ \\
Computer Science Department, University of California, Irvine}
\else
\Issue{0}{0}{0}{0}{0} 
\HeadingAuthor{Bannister, Eppstein, Simons} 
\HeadingTitle{Compaction for Graph Drawings} 
\Ack{This work was supported in part by NSF grant 0830403 and by the Office of Naval Research under grant N00014-08-1-1015. A previous version of this paper appeared in the International Symposium on Graph Drawing 2011.}

\author{Michael J. Bannister}{mbannist@uci.edu}
\author{David Eppstein}{eppstein@uci.edu}
\author{Joseph A. Simons}{jsimons@uci.edu}

\affiliation{Computer Science Department,\\
University of California, Irvine}

\submitted{}%
\reviewed{}%
\revised{}%
\reviewed{}%
\revised{}%
\accepted{}%
\final{}%
\published{}%
\type{Regular paper}%
\editor{}%
\fi

\maketitle

\begin{abstract}
We show that several problems of compacting orthogonal graph drawings to use the minimum number of rows, area, length of longest edge or total edge length cannot be approximated better than  within a polynomial factor of optimal in polynomial time unless $\P=\NP$. We also provide a fixed-parameter-tractable algorithm for testing whether a drawing can be compacted to a small number of rows.
\end{abstract}

\ifFull
\else
\Body 
\fi

\section{Introduction}
\emph{Orthogonal graph drawing} is a widely used graph drawing style for low-degree graphs, in which each vertex is represented as a point or a rectangle in an integer grid, and each edge is represented as a polyline composed out of axis-parallel line segments~\cite{Eiglsperger:2001}. When used for nonplanar graphs (Figure~\ref{fig:ortho-example}),
orthogonal drawing has several desirable properties including polynomial area, high angular resolution, and right-angled edge crossings; the last property, in particular, has been shown to aid in legibility of graph drawings~\cite{HuaHonEad-PacVis-08}.

\begin{figure}[b]
\centering
\ifFull
\includegraphics[scale=1]{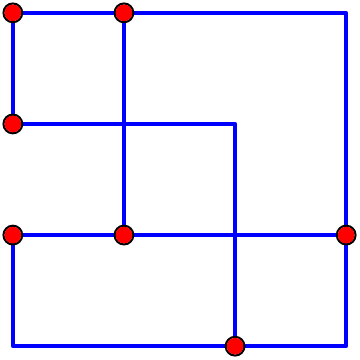}\rule{4em}{0em}
\includegraphics[scale=1]{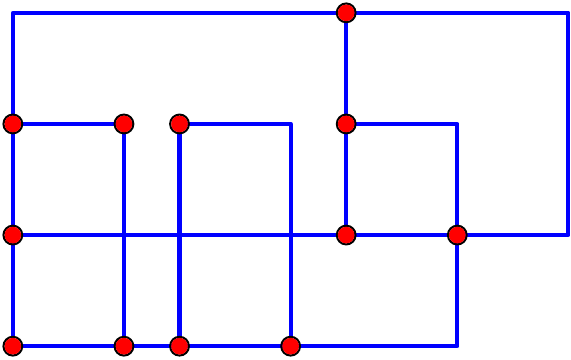}
\else
\includegraphics[scale=0.85]{figures/ortho_example1}\rule{4em}{0em}
\includegraphics[scale=0.85]{figures/ortho_example2}
\fi
\caption{Examples of non-planar orthogonal drawings.}
\label{fig:ortho-example}
\end{figure}

Typical orthogonal graph drawing systems employ a multiphase approach~\cite{BieMadTol-IJCGA-00,Eiglsperger:2001} in which the input graph is \emph{planarized} by replacing its crossings with vertices, a \emph{topological embedding} of the graph (specifying the ordering of the edges around each vertex, but not the vertex and edge locations) is found, a flow algorithm is used to orient the edges in a way that minimizes the number of bends~\cite{Tam-SJC-87}, and vertex coordinates are assigned.  If vertices of degree greater than four exist, they may be expanded to rectangles as another phase of this process~\cite{BieMadTol-IJCGA-00}. Finally, the drawing is improved by \emph{compaction}, a step in which the vertices and bends of the graph are moved to new locations in order to reduce the area of the drawing while preserving its edge orientations and other features.

Some positive algorithmic results are known for the final compaction step for planar drawings;
for instance, Bridgeman et al.~\cite{BriDiBDid-CGTA-00} showed that planar orthogonal drawings in which the shapes of the faces in the drawing are restricted (so-called \emph{turn-regular drawings}) may be compacted into optimal area in polynomial time.
In the case of variable size vertices, e.g. due to labels, Eiglsperger et al. \cite{Eiglsperger-VarVert} have found heuristic based compaction algorithms that work well in practice. Also, Klau et al. have shown that the compaction step can be done via integer linear programming \cite{Klau-Zlp}. An experimental analysis of planar compaction methods was done by Klau et al. \cite{Klau-Exp}

However, when drawing nonplanar graphs, it may not be necessary or desirable for the compaction phase to preserve a fixed planarization of the graph. If one is compacting one dimension of a drawing at a time, then for planar compaction it is only possible to map the rows of the drawing monotonically to a smaller set of rows, while for nonplanar graphs it may also be useful to permute the rows with respect to each other. This greater freedom to choose how to compact the drawing may lead to much greater savings in drawing area, but it also leads to greater difficulty in finding a good compaction.

As Patrignani~\cite{Patrignani:1999} showed, even for arbitrary planar orthogonal graph drawings, compacting the drawing in a way that minimizes its area, total edge length, or maximum edge length is $\NP$-hard. Although these results do not directly extend to the nonplanar case, $\NP$-hardness in that case also follows from results of Eades et al. on rectilinear (bendless) drawing~\cite{Eades:2010}, and Ma{\v{n}}uch et al. where certain restricted cases of rectilinear drawing are considered \cite{Manuch:2011}.
But since compaction is performed primarily for aesthetic reasons (a smaller area drawing allows the drawing to be viewed at a larger scale, making its features more legible), exact optimization may not be important if a layout with small area can be achieved. Thus, we are led to the problem of how closely it is possible to \emph{approximate} the minimum area layout.
The problem of approximate compaction for nonplanar orthogonal drawings was explicitly listed as open by Eiglsperger et al.~\cite{Eiglsperger:2001}, and there appears to have been little progress on it since then.

In this paper we show that nonplanar compaction is hard even to approximate.
There exists a real number $c>0$ such that, unless $\P=\NP$, no polynomial time
algorithm can find a compaction of a drawing with $n$ features that is within a
factor of $n^c$ of optimal. The main idea is to find approximation-preserving
reductions from graph coloring, a problem known to be hard to approximate. We
also give a \emph{fixed-parameter tractable} algorithm for finding compactions 
that use very small numbers of grid rows, for drawings in which such a compaction is possible.

\subsection{Variations of the compaction problem}

In the compaction problems we study, the task is to move vertices and bends
while preserving the axis-parallel orientation (although not necessarily the
direction) of each edge, to minimize the number of rows, area, longest edge,
or total edge length of the
drawing. Our results apply either to \emph{orthogonal drawings} (drawings in
which edges may be polylines with bends, possible for any graph of maximum
degree four) or \emph{rectilinear drawings} (bendless drawings,  only
possible for some graphs)~\cite{Eades:2010,Eppstein:2009}; the distinction
between bends and vertices is unimportant for our results.

We distinguish between three variants of the compaction problem, depending on what vertex motions are allowed. In \emph{row-by-row compaction} (Figure~\ref{fig:rxreg}), the compacted layout maps each row of the input layout to a row of the output; all vertices that belong to the same row must move in tandem. In \emph{vertex-by-vertex vertical compaction} (Figure~\ref{fig:vxveg}), each vertex or bend may move independently, but only its $y$-coordinate may change; it must retain its horizontal position. In \emph{vertex-by-vertex free compaction}, vertices or bends may move arbitrarily in both coordinate directions. In all three of these problems, edges or edge segments must stay vertical or horizontal according to their orientation in the original layout. The compaction is not allowed to cause any new intersection between a vertex and a feature it was not already incident with, nor is it allowed to cause any two edges or edge segments to overlap for nonzero length; however, it may introduce new crossings that were not previously present. Therefore, the compacted drawing remains a valid representation of the graph; however, the introduction of crossings may reduce the aesthetics and readability of the graph, especially in cases where the initial input drawing is planar.

\begin{figure}[ht!]
\centering
\ifFull
\includegraphics[scale=1]{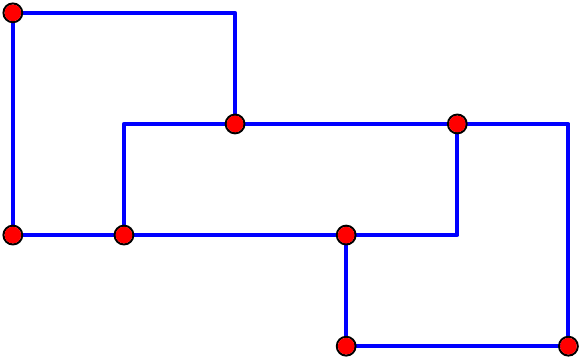}\rule{4em}{0em}
\includegraphics[scale=1]{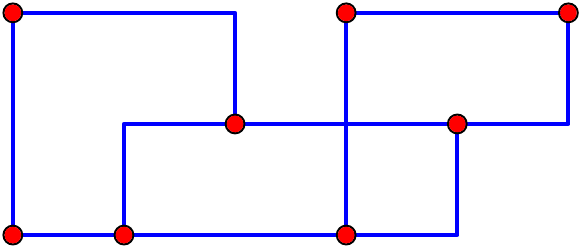}
\else
\includegraphics[scale=0.8]{figures/lcomp_example_before}\rule{4em}{0em}
\includegraphics[scale=0.8]{figures/lcomp_example_after}
\fi
\caption{Example of row-by-row compaction.}
\label{fig:rxreg}
\end{figure}

\begin{figure}[ht!]
\centering
\ifFull
\includegraphics[scale=1]{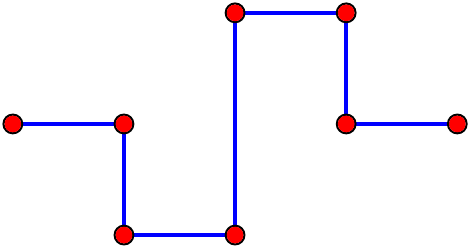}\rule{4em}{0em}
\includegraphics[scale=1]{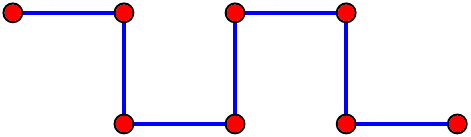}
\else
\includegraphics[scale=0.8]{figures/hscomp_example_before}\rule{4em}{0em}
\includegraphics[scale=0.8]{figures/hscomp_example_after}
\fi
\caption{Example of vertex-by-vertex compaction.}
\label{fig:vxveg}
\end{figure}

\subsection{New results}

We show the following results.
\begin{itemize}
\item 
In the row-by-row compaction problem, it is difficult to compact even a
drawing of a path graph (or a drawing of the two-vertex graph with many bends).
If the drawing has $n$ vertices or bends, then for every $\epsilon>0$  there is no
polynomial time algorithm that can find a compacted drawing whose number of
rows, maximum edge length, total edge length, or area is within
$O(n^{1/2-\epsilon})$ of optimal, unless $\P=\NP$. Moreover, even
finding drawings with a fixed number of rows is hard; it is $\NP$-complete
to determine whether there exists a compaction with only three rows.
\item 
In vertex-by-vertex vertical compaction, for every $\epsilon>0$, there is no
polynomial time algorithm that can find a compacted drawing of graphs of maximum degree three such that the number of
rows, maximum edge length, total edge length, or area of the drawing is within
$O(n^{1/3-\epsilon})$ of optimal, unless $\P=\NP$. The same result also applies in the
vertex-by-vertex free compaction problem.
\newpage 
\item 
For the vertex-by-vertex or vertex-by-vertex free compaction problem of  three-dimensional
orthogonal drawings, for every $\epsilon>0$, it is not possible to approximate the
minimum number of layers in any one dimension, maximum edge length, total edge length, or volume to within $O(n^{1/2-\epsilon})$ of
optimal in polynomial time, unless $\P=\NP$, nor is it possible in
polynomial time to determine whether a three-layer drawing exists.
\item 
In vertex-by-vertex vertical compaction, there is an algorithm for testing
whether an orthogonal graph drawing can be compacted into $k$ rows, whose
running time is $O(k! n)$. Thus, the problem is fixed-parameter tractable.
\item 
In vertex-by-vertex free vertical compaction, there is a simple linear time algorithm for finding an optimal row compaction of a path.
\item
We provide a counterexample showing that an approximation algorithm described in the conference version of this paper was incorrect.
\end{itemize}

\section{Preliminaries}
\subsection{Orthogonal Drawing}

We define an \emph{orthogonal drawing} of a graph to be a drawing in which each vertex is represented as a point in the Euclidean plane (although most of our results apply as well to drawings in which the vertices are rectangles), and each edge is represented as a polyline (a polygonal chain of line segments), with each line segment parallel to one of the coordinate axes.
(See, for example, Figure~\ref{fig:ortho-example}.)
If each edge is itself a line segment, the drawing is \emph{rectilinear}; otherwise, the segments of a polyline meet at \emph{bends}. Each vertex or bend must only intersect the edges that it belongs to, and no two vertices or bends may coincide. Edges may cross each other, but only at right angles, at points that are neither vertices nor bends.

It is natural, in orthogonal graph drawing, to restrict the coordinates of the vertices and bends to be integers. In this case, the \emph{width} of a two-dimensional drawing is the maximum difference between the $x$-coordinates of any two of its vertices or bends, the \emph{height} is the maximum difference between $y$-coordinates of any two vertices or bends, and the \emph{area} is the product of the width and height.

A \emph{compaction} of a drawing $D$ is another drawing $D'$ of the same graph, in which the vertices and bends of $D'$ correspond one-for-one with the vertices and bends of $D$, and in which corresponding segments of the two drawings are parallel to each other. Typically, $D'$ will have smaller height, width, or area than $D$. We distinguish between three types of compaction:
\begin{itemize}
\item In \emph{row-by-row compaction}, the $x$-coordinate of each vertex or bend remains unchanged, and two vertices or bends that have the same $y$-coordinate in $D$ must continue to have the same $y$-coordinate in $D'$ (Figure~\ref{fig:rxreg}).
\item In \emph{vertex-by-vertex vertical compaction},  the $x$-coordinate of each vertex or bend remains unchanged, but the $y$-coordinates may vary independently of each other subject to the condition that the result remains a valid drawing with  edge segments parallel to the original drawing (Figure~\ref{fig:vxveg}).
\item In  \emph{vertex-by-vertex free vertical compaction}, the $x$- and $y$- coordinates of each vertex or bend are free to vary independently of other vertices or bends.
\end{itemize}
As can be seen in Figure~\ref{fig:rxreg}, we allow compaction to introduce new edge crossings and to reverse the directions of edge segments.
These concepts generalize straightforwardly to three dimensions.

\subsection{Graph Coloring and Inapproximability}

In the \emph{graph coloring problem}, we are given as input a graph and seek to color the vertices of the graph with as few colors as possible, in such a way that the endpoints of each edge are assigned different colors. Our results on the difficulty of compaction are based on known inapproximability results for graph coloring, one of the triumphs of the theory of probabilistically checkable proofs~\cite{Zuckerman:2007, Feige-Color, Khot-Color}.

\begin{theorem}[Zuckerman~\cite{Zuckerman:2007}]
\label{lem:chi-rho}
Let $\epsilon>0$ be any fixed constant. Then, unless $\P=\NP$, there is no
polynomial time algorithm that can compute the minimal number of colors
$\chi(G)$ in an optimal coloring of a given $n$-vertex graph within a factor of $n^{1-\epsilon}$.
\end{theorem}

Our proofs use approximation-preserving reductions from coloring to compaction.
Given a graph $G$ to be colored, we will construct a different graph $G'$ and a
drawing $D$ of $G'$ such that the number of rows, maximum edge length, or total
edge length in a compaction $D'$ of $D$ necessarily correspond to the number of
colors in a coloring of $G$. With a reduction of this type, the approximation
ratio for compacting $D$ cannot be better than the approximation ratio for
coloring $G$. However, $D$ will in general have many more vertices and bends
than the number of vertices in $G$; the size of $D$ will be at least proportional to the number of edges in $G$, which is quadratic in its number of vertices. Therefore, although the approximation ratio will remain unchanged as a number by our reduction it will be expressed as a different function of the input size.

\subsection{Notation}
We write $n_G$, $n_D$, or (where unambiguous) $n$ for the number of vertices in
a graph $G$ or drawing $D$ and $m_G$, $m_D$, or $m$ for its number of edges.
Additionally, $b_D$ stands for the number of bends in drawing $D$. When we consider
compaction problems with the objectives of minimizing the number of rows, maximum edge length, and
total edge length in a drawing $D$ we will use $\lambda(D)$, $\mu(D)$, and $\tau(D)$ to denote
respectively the number of rows, maximum edge length, and total edge length 
in an optimal compaction. $\chi(G)$ represents the chromatic number of graph $G$.

\section{Hardness of Row-By-Row Compaction}
\label{sec:row-by-row}

As a warm-up, we start with a simplified path compaction problem in which every pair of objects on the same row of the drawing must move in tandem. Our proof constructs a drawing of a path graph such that the valid row assignments for our drawing are the same as the valid colorings of a given graph~$G$.

\begin{lemma}
\label{path lemma}
Given a graph $G$ we can construct in polynomial time a rectilinear drawing $D$ of a path graph with $O(m_G)$ vertices, such that $\lambda(D) = \chi(G)$.
\end{lemma}

\begin{proof}
Find a Chinese postman walk for $G$ of length less than $2m_G$; that is, a walk that starts at an arbitrary vertex and visits each edge at least once, allowing vertices to be visited multiple times. Such a walk may be found, for instance, by doubling each edge of $G$ and constructing an Euler tour of the doubled graph. Let $u_iv_i$ be the $i$th edge in the walk, where $v_i=u_{i+1}$, and let $k\le 2m_G$ be the number of edges in the walk. Additionally, choose arbitrary distinct integer labels for the vertices of $G$  with $\ell(v)$ being the label for the vertex $v$.

To construct the drawing $D$, for $i$ from $0$ to $k$, place
vertices in the plane at the points $(i,\ell(u_i))$ and $(i+1,\ell(u_i))$, connected by a unit-length horizontal edge. Additionally, for $i$ from $0$ to $k-1$ draw a vertical edge from $(i+1,\ell(u_i))$ to $(i+1,\ell(v_i))$. See Figure \ref{fig:path_example} for an example of such a construction.

Two rows in the drawing conflict if and only
if the corresponding vertices in $G$ are adjacent. For every coloring of $G$,
we may compact $D$ by using one row for the vertices of each color, and conversely for every row-by-row compaction of $D$ we may color $G$ by using one color class for each row of the compaction (Figure~\ref{fig:row-by-row-coloring}). Therefore, $\lambda(D) = \chi(G)$. Also, $n_D=2k+2=O(m_G)$.
\end{proof}

\begin{figure}[t]
\centering
\includegraphics{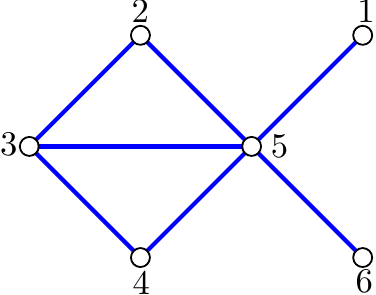}\rule{4em}{0em}
\includegraphics{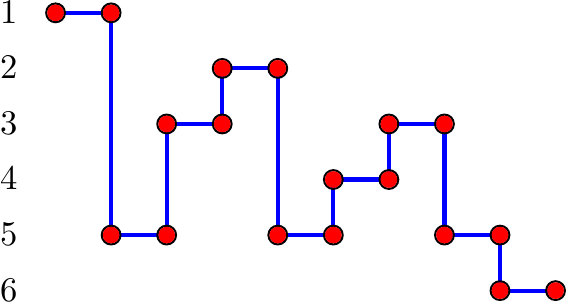}
\caption{Path constructed from a graph $G$ with the walk 1,~5,~3,~2,~5,~4,~3,~5,~6.}
\label{fig:path_example}
\end{figure}

\begin{figure}[t]
\centering
\includegraphics{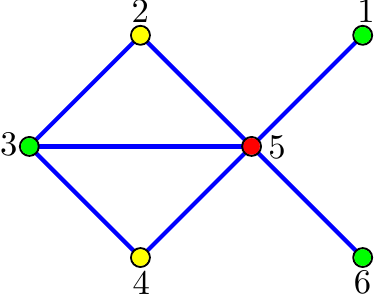}\rule{4em}{0em}
\includegraphics{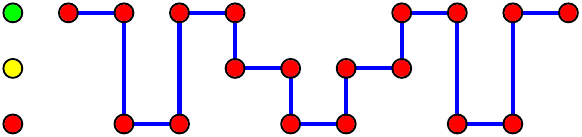}
\caption{The rows of a compacted drawing $D$ correspond to the colors in a coloring of~$G$.}
\label{fig:row-by-row-coloring}
\end{figure}

The same drawing $D$ can equivalently be viewed as an orthogonal drawing of the two-vertex graph $K_2$ with $O(m_G)$ bends. In the restricted model of compaction used in this section, horizontal compaction is disallowed, so optimizing the area of a compaction of $D$ is the same as optimizing the number of rows.

\begin{theorem}
\label{thm:row-by-row}
Let $\epsilon>0$ be any positive fixed constant, and suppose that $\P \neq \NP$.
Then there does not exist a polynomial time algorithm that
approximates the number of layers or the area in an optimal row-by-row compaction of a
given orthogonal or rectilinear drawing $D$ to within a factor of $(n_D+b_D)^{1/2-\epsilon}$. 
\end{theorem}

\begin{proof}
Suppose for a contradiction that algorithm $\mathcal{A}$ can solve the row-by-row compaction problem to within a factor $\rho \le n_D^{1/2-\epsilon}$ of optimal. Let $\mathcal{A}'$ be an algorithm for  coloring an input graph $G$ by performing the following steps:
\begin{enumerate}
\item Use Lemma~\ref{path lemma} to construct a path drawing $D$ from the given graph $G$.
\item Use algorithm $\mathcal{A}$ to compact $D$.
\item Color $G$ using one color for each row of the compacted drawing.
\end{enumerate}
Then the approximation ratio of algorithm $\mathcal{A}'$ for coloring is the same number $\rho$ as the approximation ratio of algorithm $\mathcal{A}$ for compaction, whether measured by area or by number of rows. However,
$$\rho\le n_D^{1/2-\epsilon}=O(m_G^{1/2-\epsilon})=O(n_G^{1-2\epsilon}),$$
an approximation ratio that contradicts Theorem~\ref{lem:chi-rho}.
\end{proof}

The same reduction, together with the $\NP$-completeness of graph 3-color\-ability, shows that it is $\NP$-complete to determine whether a given drawing $D$ has a row-by-row compaction that uses at most three rows.

\begin{theorem}
The problem of determining if an orthogonal graph has a row-by-row compaction into at most three rows is $\NP$-complete.
\end{theorem}

\begin{figure}[ht]
\centering
\ifFull
\includegraphics[width=0.7\textwidth]{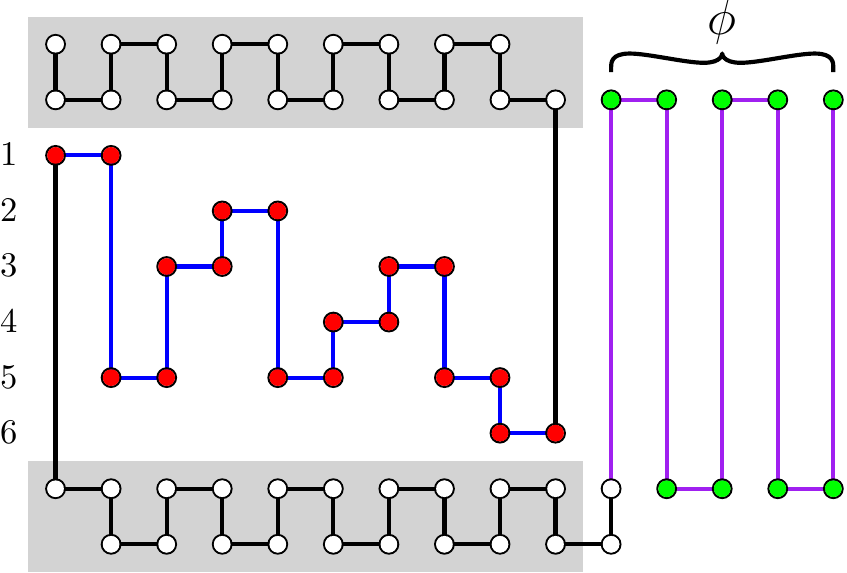}
\else
\includegraphics[width=0.8\textwidth]{figures/longEdge_row}
\fi
\caption{
\label{fig:path_example_edgeLength}
We add a frame above and below the path construction of Figure~\ref{fig:path_example}
(shaded gray). The longest edge runs along the right side of the drawing from the top
frame to the bottom frame and will have length exactly $\chi + 1$.}
\end{figure}

For the compaction problems with respect to longest edge, and total edge length we will modify the construction in Lemma~\ref{path lemma} by the addition of a frame that will be used to measure $\chi(G)$ with $\tau(D)$ and $\mu(D)$. The frame is illustrated in Figure~\ref{fig:path_example_edgeLength}.

\begin{lemma}
\label{lem:rxrframe}
Given a graph $G$ we can construct in polynomial time a rectilinear drawing $D$ of a path graph with $O(m_G+\phi)$ vertices, such that
\[
\mu(D) = \chi(G) + 1 \quad\text{and}\quad
\phi \chi(G) \leq \tau(D) \leq O(m_G + \phi) \chi(G).
\]
\end{lemma}
\begin{proof}
Given a graph $G$ we construct $D$ by first performing the construction of Lemma~\ref{path lemma}, and then adding the vertices and edges as illustrated in Figure~\ref{fig:path_example_edgeLength}.

We claim that vertices created by Lemma~\ref{path lemma} must stay between the two gray regions in a row-by-row compaction. The leftmost vertex cannot be placed above the upper gray region without causing an illegal edge overlap, and if one vertex is placed above the upper gray region then all vertices must be placed above the upper gray region to avoid illegal edge overlaps. Hence none of the original vertices can be placed above the upper gray region (unless the entire drawing is flipped). Similar reasoning shows that none of the original vertices may be placed below the lower gray region.

By construction the $\phi$ edges to the right of the drawing are the longest edges in the drawing, even after a row-by-row compaction. The length of one of these edges is always $2$ more than the length of the gap between the gray regions, which is $\chi(G) - 1$ by Lemma~\ref{path lemma}. Hence the length of the edges in an optimal compaction is $\mu(D) = \chi(G) + 1$.

For the total edge length we note that there are at most $2m_G + \phi$ columns, and each of these columns has at most $2\chi(G)$ in length. The horizontal segments make at most $O(m_G + \phi)$ total length. So altogether we have that $\tau(D)$ is at most $O(m_G + \phi) \chi(G)$. Considering only the vertical edges on the right we see that $\phi \chi(G) \leq \tau(D)$, which gives the desired bounds.
\end{proof}

\begin{theorem}
\label{thm:row-by-row-edgeSum}
Let $\epsilon>0$ be any positive fixed constant, and suppose that
$\P\neq\nobreak\NP$.
Then there does not exist a polynomial time algorithm that
approximates the longest edge length or the total edge length in an optimal row-by-row compaction of a given orthogonal or rectilinear drawing $D$ to within a factor of $(n_D+b_D)^{1/2-\epsilon}$. 
\end{theorem}
\begin{proof}
The proof of the longest edge case is the same as the proof of Theorem~\ref{thm:row-by-row} for the number of rows where Lemma~\ref{lem:rxrframe} is used in place of Lemma~\ref{path lemma}.

More work is needed for the total length case. Suppose for a contradiction that an algorithm
$\mathcal{A}$ could solve the row-by-row compaction problem with respect to total
edge length to within a factor of $n_D^{1/2-\epsilon}$ of optimal. Let
$\mathcal{A}'$ be an approximation algorithm for $\chi(G)$ that compacts the drawing produced by Lemma~\ref{lem:rxrframe} with $\phi = n_G^2$, and returns $\tau'(D)/n^2_G$ where $\tau'$ is the total edge length of the approximate compaction.

From Lemma~\ref{lem:rxrframe} we know that
\[
n^2_G \chi(G) \leq \tau(D) \leq O(n^2_G)\chi(G),
\quad\text{and}\quad
\tau(D) \leq \tau'(D) \leq n_D^{1/2 - \epsilon}\tau(D)
\]
by assumption. Thus,
\[
\chi(G) \leq \tau'(D)/n_G^2 \leq n_D^{1/2-\epsilon}/n_G^2 \tau(D)
\leq O(n_G^2)/n_G^2 n_D^{1/2 - \epsilon} \chi(G)
\leq O(n_G^{1-2\epsilon}) \chi(G)
\]
contradicting Theorem~\ref{lem:chi-rho}.
\end{proof}

\section{Hardness of Vertex-By-Vertex Compaction}
\label{section:vertex-vertex}

Our hardness result for vertex-by-vertex vertical compaction follows roughly the same outline as Theorem~\ref{thm:row-by-row}: translate graph vertices into drawing features such that two features  can be compacted onto the same row if and only if the corresponding graph vertices can be assigned the same color. However, if we translated adjacencies in the graph to be colored into direct overlaps between pairs of features, as we did for row-by-row compaction of path graphs, this translation would only let us represent interval graphs, which are easily colored~\cite{Olariu:1991}. Instead we use an \emph{edge gadget} depicted in Figure~\ref{fig:edge-gadget}
\begin{figure}[b]
\centering
\ifFull
\includegraphics[scale=1.25]{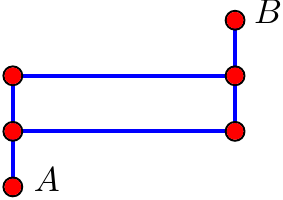}\qquad
\includegraphics[scale=1.25]{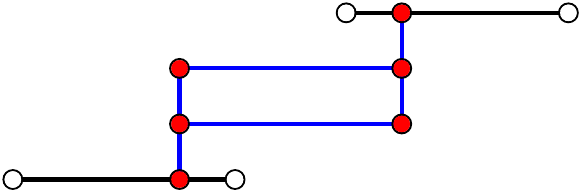}
\else
\includegraphics{figures/cgadget}\qquad
\includegraphics{figures/connected_gadget}
\fi
\caption{The basic version of an edge gadget (left), connecting two horizontal segments representing vertices (right).}
\label{fig:edge-gadget}
\end{figure}
to represent an edge between two vertices of the
input graph. This gadget has six vertices and six line segments; the two vertices $A$ and $B$ of the gadget may be placed on two line segments representing vertices of the input graph. This connection forces the two line segments containing $A$ and $B$ to be placed on different rows of any compacted drawing, even if these two line segments have no vertical overlap with each other: one of the two line segments must be above the central rectangle of the gadget, and the other must be below the central rectangle, although either of these two orientations is possible.

\begin{figure}[b]
\centering
\ifFull
\includegraphics[scale=0.7]{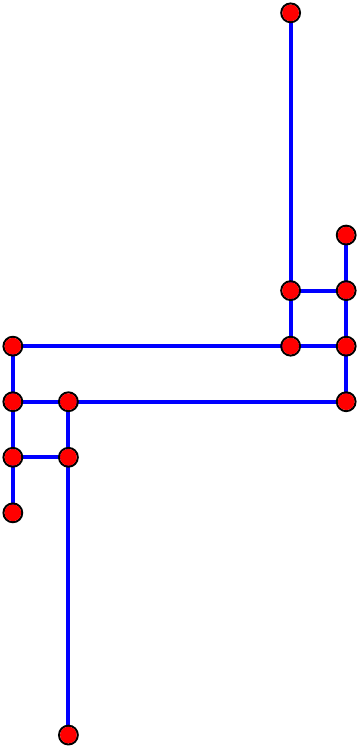}\rule{2em}{0em}
\includegraphics[scale=0.7]{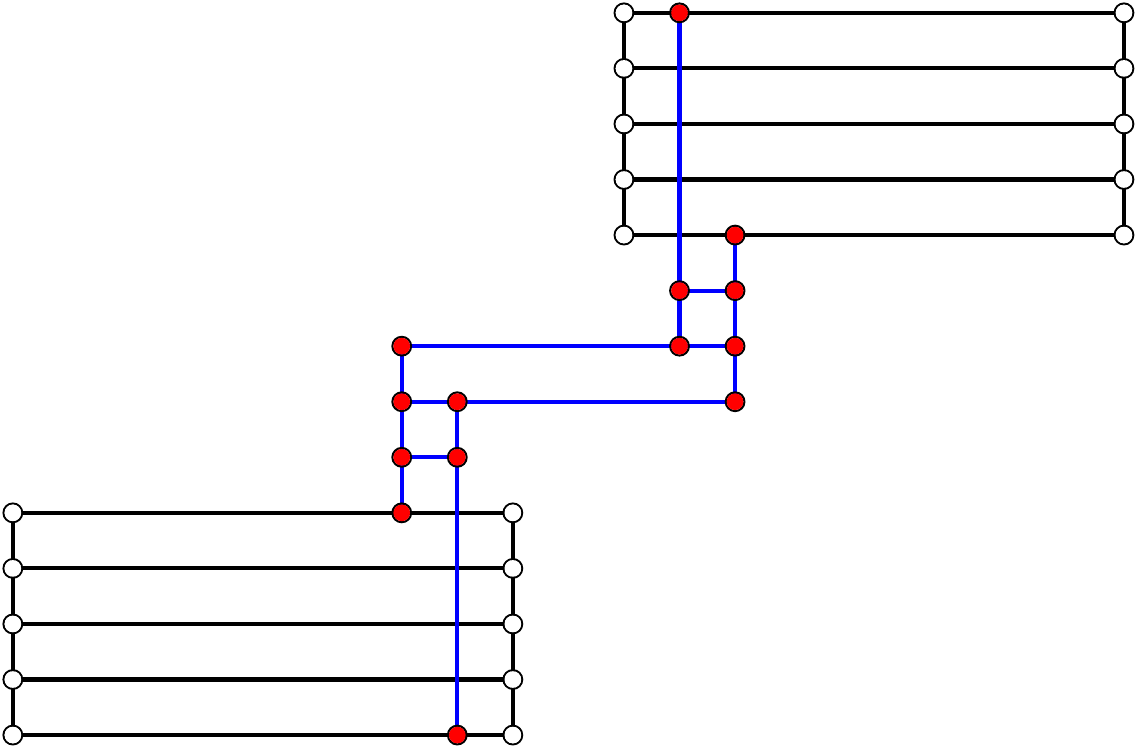}
\else
\includegraphics[scale=0.65]{figures/mcgadget_new}\rule{2em}{0em}
\includegraphics[scale=0.65]{figures/kconnected_gadget_new}
\fi
\caption{The full edge gadget for $\theta=5$.}
\label{fig:bundle-gadget}
\end{figure}

The use of these edge gadgets leads to a second difficulty in our reduction: the number of rows in the compacted drawing will depend both on the features coming from input graph vertices and the rows needed by the edge gadgets themselves. In order to make the first of these two terms dominate the total, we represent an input graph vertex by a \emph{vertex bundle} of $\theta$ parallel line segments, for some integer $\theta>0$. The edge gadgets may be modified to enforce that all segments in one bundle be in different rows from all segments of a second bundle, as shown in Figure~\ref{fig:bundle-gadget}, while only using a constant number of rows for the gadget itself.

\begin{figure}[t]
\centering
\includegraphics[width=0.275\textwidth]{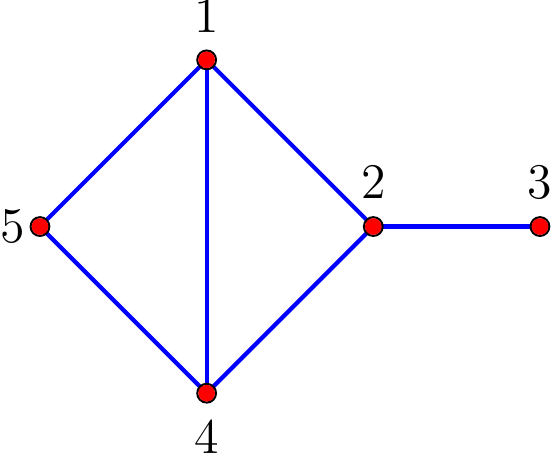}\hfill
\includegraphics[width=0.7\textwidth]{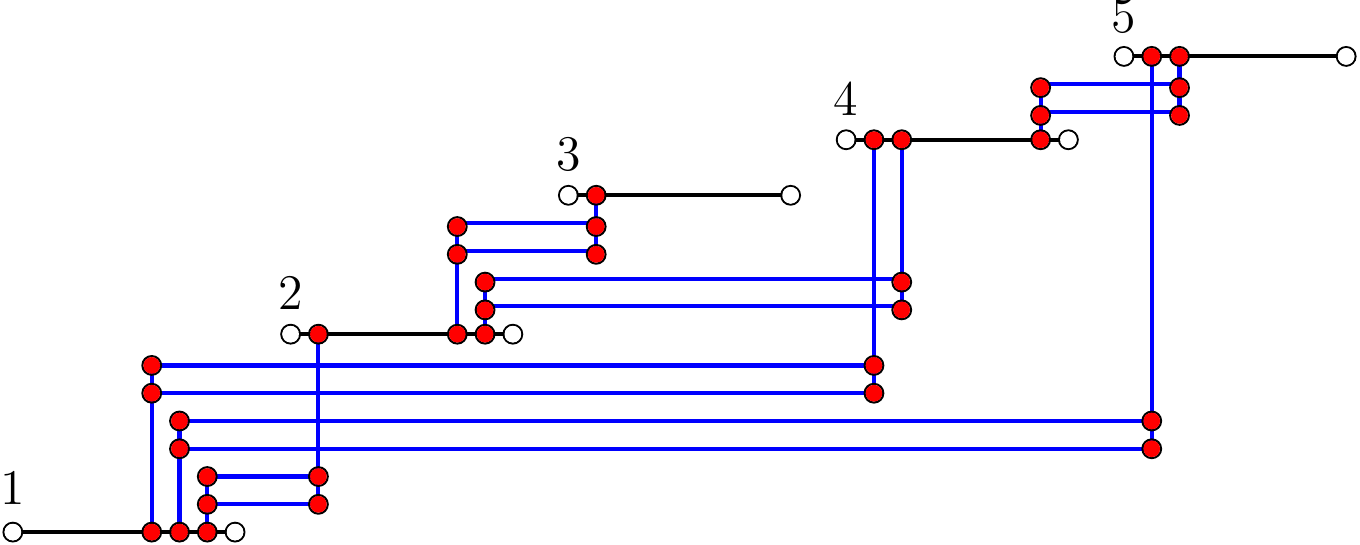}
\caption{Example of the complete reduction for $\theta=1$.}
\label{fig:vvredux}
\end{figure}

\begin{figure}[t]
\centering
\includegraphics[width=0.275\textwidth]{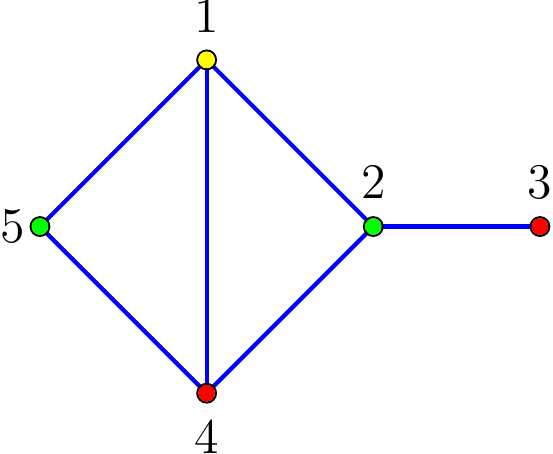}\hfill
\includegraphics[width=0.7\textwidth]{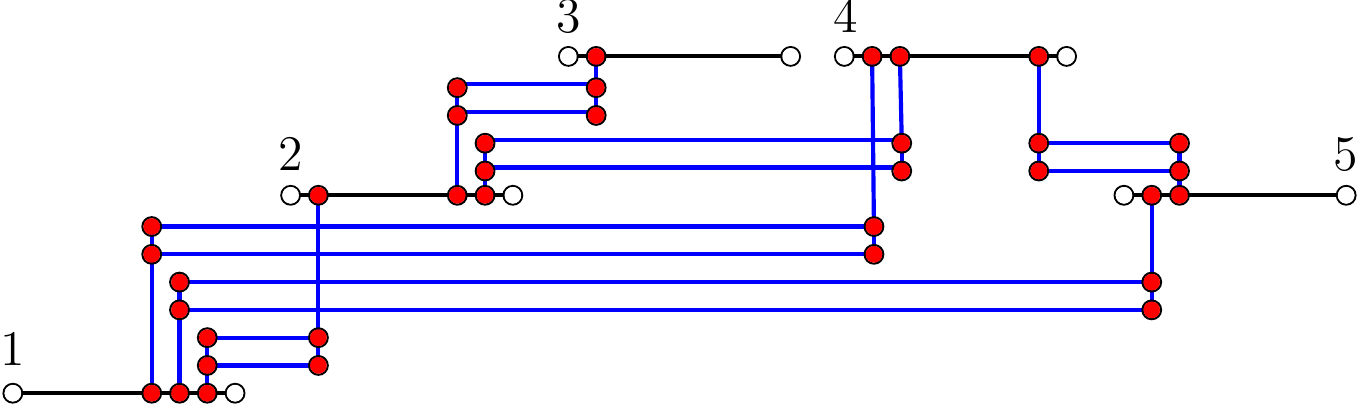}
\caption{Example coloring based on the compaction of the drawing in Figure~\ref{fig:vvredux}.}
\label{fig:vvredux-c}
\end{figure}

Figure~\ref{fig:vvredux} shows the complete reduction, for a graph $G$ with five vertices and six edges, and for $\theta=1$. Each vertex of $G$ is represented as a horizontal black line segment (or bundle of segments, for $\theta>1$), and each edge of $G$ is represented by an edge gadget. The vertices of $G$ are numbered arbitrarily from $1$ to $n_G$, and these numbers are used to assign vertical positions to the corresponding bundles of segments in the drawing. The edge gadgets are given $x$-coordinates that allow them to attach to the two vertex bundles they should be attached to, and $y$-coordinates that place them between these two vertex bundles.

\begin{lemma}
\label{unrestricted lemma}
Given a graph $G$ and a parameter $\theta$ we can construct in polynomial time an orthogonal drawing $D$ such that the vertices of $D$ have maximum degree~3, $n_D = O(\theta n_G + m_G)$, and
$$\theta \chi(G)\le \lambda(D) \le \theta \chi(G)+O(n_G^2).$$
\end{lemma}
\begin{proof}
The construction of $D$ is as described above. Since each edge gadget has a constant number of vertices together they contribute $O(m_G)$ to $n_D$, and each vertex bundle has $O(\theta)$ vertices together they contribute $O(\theta n_G)$ to $n_D$. Hence $n_D = O(\theta n_G + m_G)$. The bound on the degree is clear from the construction. If $G$ has a coloring with $\chi$ colors, it is possible to assign the vertex bundles of $D$ to $\chi$ sets of $\theta$ rows each, according to those colors, with an additional $O(n_G)$ rows between any two such sets to allow room for the edge gadgets to be placed without interference with each other; see Figure~\ref{fig:vvredux-c}. Therefore, $ \lambda(D)\le \theta \chi(G)+O(n_G^2)$.
If $D'$ is a compacted drawing of $D$, acyclically orient the edges of $G$ from the vertex whose bundle is below the edge gadget to the vertex whose bundle is above the edge gadget, and assign each vertex $v$ in $G$ a color indexed by the length of the longest path from a source to $v$ in this acyclic orientation. Then the number of colors needed equals the number of vertices in the longest path, and the number of rows in $D'$ needed just for the vertices in this path is $\theta$ times the number of vertices of $G$ in the path. Therefore, $\theta\chi(G) \le \lambda(D)$.
\end{proof}
\begin{theorem}
If $\P \neq \NP$, then no polynomial time algorithm approximates the number of
rows or the area in an optimal vertex-by-vertex vertical compaction of a given orthogonal graph drawing to within a factor of $n_D^{1/3 - \epsilon}$.
\end{theorem}

\begin{proof}
If an algorithm could achieve this approximation ratio for compaction, we could get an $O(n_G^{1-3\epsilon})$ ratio for coloring by applying Lemma~\ref{unrestricted lemma} with \mbox{$\theta=n_G^2$}, compacting the resulting drawing, and using the coloring derived from the compaction in the proof of Lemma~\ref{unrestricted lemma}, contradicting Theorem~\ref{lem:chi-rho}.
\end{proof}
\newpage
\begin{lemma}
\label{lem:vxvlong}
Given a graph $G$ and parameters $\theta$ and $\phi$ we can construct in polynomial time an orthogonal drawing $D$ such that the vertices of $D$ have maximum degree $3$, $n_D = O(\theta n_G + m_G + \phi)$, with the bounds
\[
\theta \chi(G) \leq \mu(D) \leq \theta\chi(G) + O(n_G^2) \quad\text{when}\quad \theta > 6n_G^2
\]
on the longest edge, and the bounds
\[
\tau(D) = \Theta(n_G^4) \chi(G)
\quad\text{when}\quad \theta,\phi = \Theta(n_G^2)
\]
on the total edge length.
\end{lemma}
\begin{proof}
We add two rows of vertices one above and one below the drawing produced by
Lemma~\ref{unrestricted lemma} (illustrated in Figure~\ref{fig:vvredux}) by
placing vertices above and below the left most column of vertices in each vertex bundle.
We also place $\phi$ additional vertices on the top and the bottom to the right of the original drawing. The vertices are then connected as illustrated by Figure~\ref{fig:longEdge_VxV}.

First, we consider the bounds on the longest edge. For these bounds we let $\phi = 1$ and $\theta > 6n^2$. The width of the original drawing (from Lemma~\ref{unrestricted lemma}) is $2n_G$ for the vertex bundle and $4m_G$ for connecting the edge gadgets to the vertex bundles; see the right side of Figure~\ref{fig:bundle-gadget}. So the width of the original drawing is bounded above by $6n_G^2$. When adding the additional frame (the green vertices with purple edge in Figure~\ref{fig:longEdge_VxV}), all added horizontal edges have length at most $2n_G$, more precisely twice the degree of the vertex being represented by the bundle. Now, by construction the vertical edge on far right is the longest vertical edge in the drawing, and has length at least
\(
\theta \chi(G) + O(m_G) > 6n_G^2
\)
by Lemma~\ref{unrestricted lemma}, making it the longest edge in the drawing even after compaction. So we have that $\mu(D) = \lambda(D) - 1$ for the constructed drawing, i.e., minimizing the longest edge in the new drawing is equivalent minimizing the number of rows in the original drawing. The bounds for longest edge now follow from the bounds in Lemma~\ref{unrestricted lemma}.

For the bounds on total edge length, we let $\theta, \phi = \Theta(n_G^2)$. Each vertex bundle contributes $O(\theta n_G)$ for a total of $O(\theta n_G^2$). The height of each edge gadget is $O(\theta \chi(G))$ from the vertex bundle plus $O(m_G)$ from the other edge gadgets. The horizontal contribution from each edge gadgets is $O(n_G^2)$, as the width of the drawing is at most $6n_G^2$. Altogether each edge gadget contributes $O(\theta\chi(G) + n_G^2)$ for a total of $O(n_G^2 \theta \chi(G) + n_G^4)$. The added frame contributes $O(n_G^2 + \phi)$ in horizontal edge length, and $\phi\theta\chi(G)$ in vertical edge length. Since $\phi = \Theta(n_G^2)$ the vertical lines holding the the drawing in the frame do not contribute to the asymptotic complexity. Summing up the total edge length,
\[
\tau(D) = O(\theta n_G^2) + O(n_G^2\theta \chi(G) + n_G^4) + \phi\theta\chi(G) = \Theta(n_G^4) \chi(G),
\]
we get the desired bounds.
\end{proof}

\begin{figure}[t]
\centering
\includegraphics[width=\textwidth]{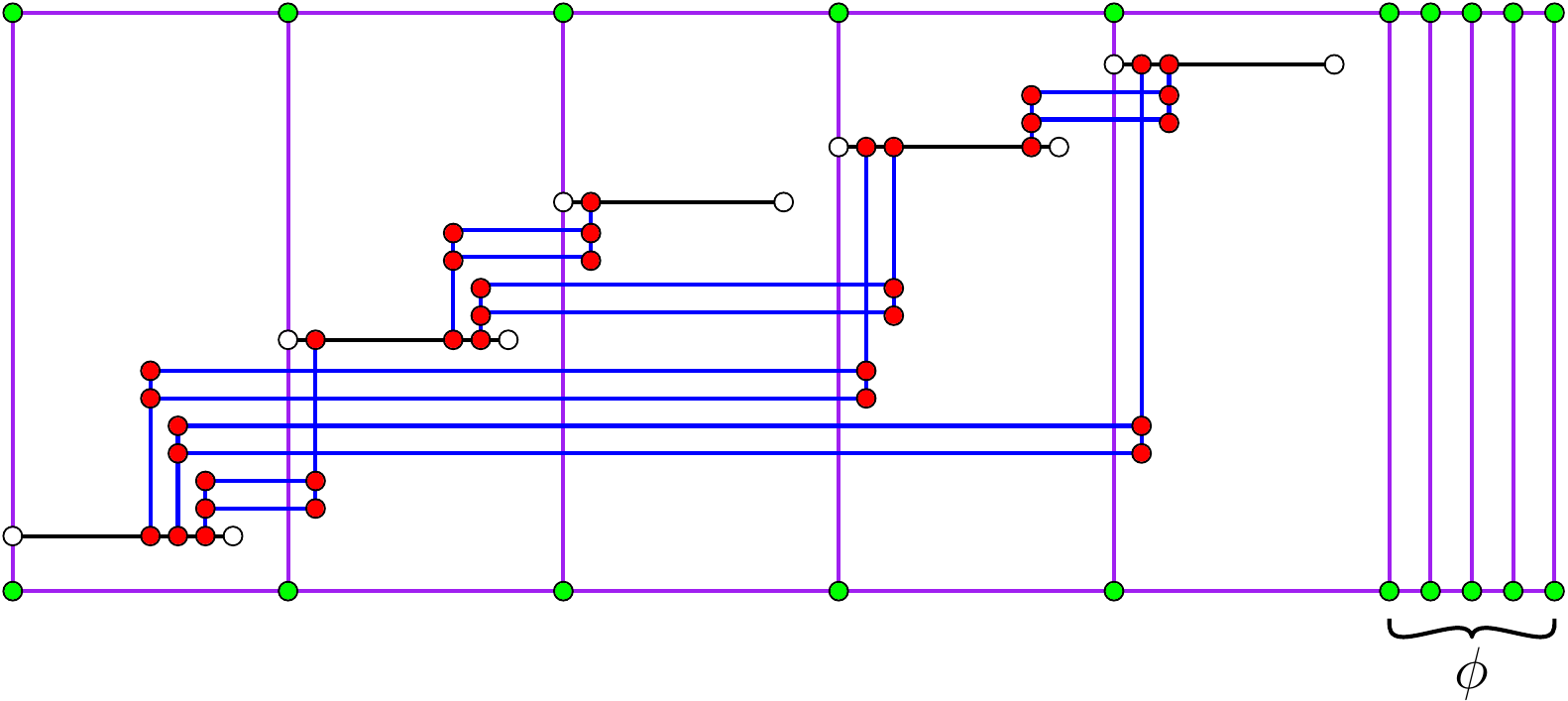}
\caption{
\label{fig:longEdge_VxV}
Long-edges-frame added to the drawing in Figure~\ref{fig:vvredux}}
\end{figure}

\begin{theorem}
If $\P \neq \NP$, then no polynomial time algorithm approximates the length of the longest edge or total edge length in an optimal vertex-by-vertex vertical compaction of a given orthogonal graph drawing to within a factor of $n_D^{1/3 - \epsilon}$.
\end{theorem}
\begin{proof}
If an algorithm could achieve this approximation ratio for compaction, we could get an $O(n^{1-3\epsilon})$ ratio for coloring by applying Lemma~\ref{lem:vxvlong} with $\theta = 7n^2$ and $\phi = 1$ and compacting with respect to longest edge, or by applying Lemma~\ref{lem:vxvlong} with $\theta = \phi = n^2$ and compacting with respect to total edge length.
\end{proof}

\section{Hardness of Vertex-By-Vertex\\ Free Compaction}

In the reduction from the previous section, allowing the vertices to move horizontally as well as vertically does not make any difference in how much vertical compaction is possible. However, if we want to prove inapproximability for minimal-area compaction, we also need to worry about horizontal compaction.
By making the width uncompactable we may make the vertical compaction factor the same
as the area compaction factor.

\begin{lemma}
\label{lem:fixbar}
From a drawing $D$ a drawing $D'$ can be constructed by adding at most
$O(n_D)$ vertices, such that $\lambda(D') = \lambda(D) + 1$ and $D'$ is
uncompactable in the horizontal direction. If $D$ has maximum degree three, then so does $D'$.
\end{lemma}

\begin{figure}[t]
\centering
\ifFull
\includegraphics[scale=1.25]{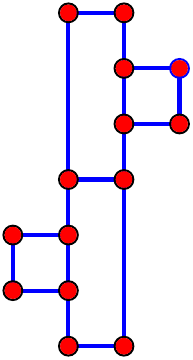}\rule{2em}{0em}
\includegraphics[scale=1.25]{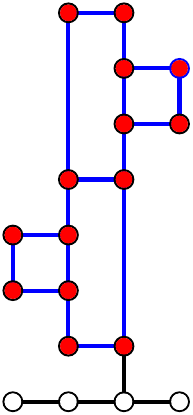}
\else
\includegraphics{figures/area_example}\rule{2em}{0em}
\includegraphics{figures/area_example_bar}
\fi
\caption{Adding a row of vertices to $D$ prevents  horizontal compaction.}
\label{fig:uncompactable}
\end{figure}
\begin{proof}
Place a line of vertices on a new row below $D$; for each set of vertices with a given $x$-coordinate in $D$, add a vertex on the new row at the
same $x$-coordinate. Connect the added
vertices with horizontal edges, and add a vertical edge to connect these vertices to $D$ at the point of $D$ that is rightmost on its bottom row, as shown in Figure~\ref{fig:uncompactable}. This added layer conflicts with all
existing horizontal layers, and forces $D'$ to be uncompactable in the
horizontal direction.
\end{proof}

\begin{theorem}
Unless $\P= \NP$, it is impossible to find vertex-by-vertex free compactions with respect to number of rows, area, longest edge length or total edge length within a factor of $n_D^{1/3-\epsilon}$ of optimal in polynomial time.
\end{theorem}
\begin{proof}
For the case of compaction with respect to area we apply Lemma~\ref{lem:fixbar} to reduce compaction with respect to area to compaction with respect to the number of rows. Since we only added a linear number of vertices the approximation ratios are the same. The rest of the theorem follows from the results on vertex-by-vertex vertical compaction.
\end{proof}

\section{Hardness of Three-Dimensional Compaction}
Our hardness result for three-dimensional vertical compaction (along the $z$-axis) follows from the construction of a drawing whose valid two-dimensional layer assignments are the same as the valid colorings of a graph $G$. We assign to each vertex in $G$ a horizontal layer parallel to the $xy$-plane containing an $L$-shaped pair of line segments, such that when projected vertically onto the $xy$-plane every two of these $L$-shaped gadgets cross each other. For each edge in $G$ we place a vertical edge in the drawing connecting the $L$-shaped gadgets that correspond to the endpoints of the edge. Figure \ref{fig:3dredux} shows an example.
\begin{figure}[b]
\centering
\ifFull
\includegraphics[scale=1.2]{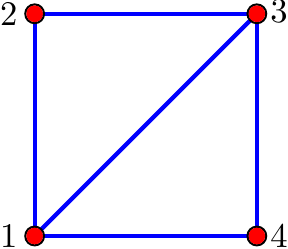}\rule{4em}{0em}
\includegraphics[scale=1]{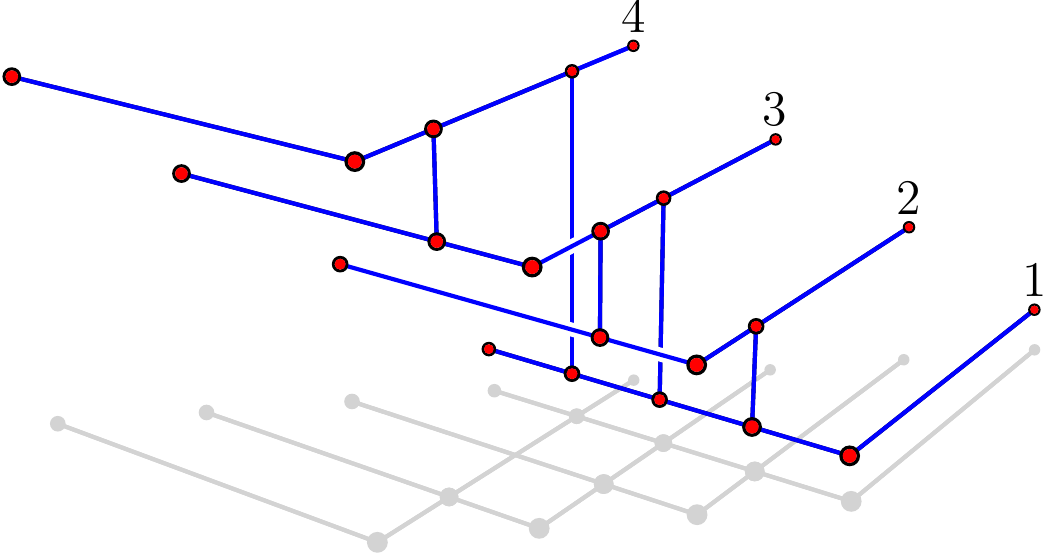}
\else
\includegraphics[scale=1]{figures/3d_example_graph}\rule{4em}{0em}
\includegraphics[scale=0.7]{figures/3d_example_shadow2}
\fi
\caption{Reduction from coloring to three-dimensional compaction where $y$ is the vertical direction.}
\label{fig:3dredux}
\end{figure}

\begin{lemma}
\label{lem:3dconst}
Given a graph $G$ we can construct in polynomial time a $3D$ orthogonal drawing $D$ with maximum degree three such that $n_D = 3n_G+2m_G=O(n_G^2)$, and such that the number of layers in an optimal $y$-compaction is $\chi(G)$.
\end{lemma}

\begin{proof}
Our construction uses three vertices in $D$ to form the $L$-shaped gadget for each vertex in $G$, and two more vertices for the vertical edge corresponding to each edge in $G$, so $n_D = 3n_G + 2m_G = O(n_G^2)$. By construction the only conflicting features in a vertical (along the $z$-axis) compaction come from the vertical edges, forcing $\chi(G)$ to be the number of layers in an optimal compaction of $D$.
\end{proof}

\begin{theorem}
If $\P \neq \NP$, then there does not exist a polynomial time algorithm that
approximates the number of layers in an optimal vertex-by-vertex or vertex-by-vertex free vertical compaction of a given
three dimensional orthogonal drawing to within a factor of $n_D^{1/2-\epsilon}$.
\end{theorem}
\begin{proof}
In this case $n_D = O(n_G^2)$, and the number of layers in an optimal compaction is equal to $\chi(G)$. The proof follows same lines as the proofs in Section~\ref{sec:row-by-row} on row-by-row-compaction.
\end{proof}

\begin{figure}[b!]
\centering
\raisebox{2.1em}{\includegraphics[width=0.375\textwidth]{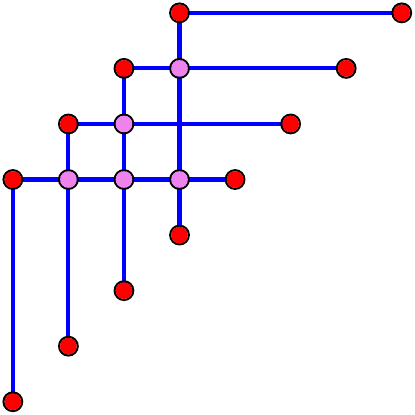}}\rule{2em}{0em}
\includegraphics[width=0.375\textwidth]{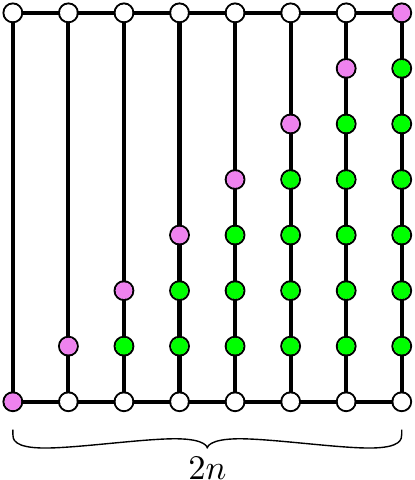}
\caption{Left: View of the 3D-drawing in Figure~\ref{fig:3dredux} from above. Vertical
edges meet at lighter shaded vertices. Right: Top/bottom frame.
}
\label{fig:3dreduxTopDown}
\end{figure}
As we did in previous sections we will reduce compaction with respect to volume,
longest edge, and total edge length to compaction with respect to layering via a frame. 
In this case we will place a mesh layer above and below the drawing, keeping the original drawing between the two layers.

\begin{lemma}
\label{lem:3dframe}
Given a graph $G$ with more than $4$ vertices we can construct in polynomial time a $3D$ orthogonal drawing $D$ with maximum degree three with $n_D = O(n_G^2)$, and such that
\[
\mu_v(D) = \chi(G) + 1 \quad\text{and}\quad n_G^2 \chi(G) \leq \tau(D) \leq O(n_G^2) \chi(G)
\]
where $\mu_v(D)$ is the length of the longest vertical edge in an optimal vertical compaction.
\end{lemma}
\begin{proof}
From the drawing produced by Lemma~\ref{lem:3dconst} we add $O(n_G^2)$ edges, and consequently $O(n_G^2)$ vertices, in $2n_G$ parallel paths above and below the drawing as illustrated in Figure~\ref{fig:3dreduxTopDown}, where the frame on the right side of the drawing is placed directly over and below the $3D$ drawing to the left. The diagonal nodes are connected to the ends of the L-gadgets from the top and the bottom to keep them between the two frames. The sub-diagonal green vertices are connected with edges between the upper and lower frame. This construction is illustrated in Figure~\ref{fig:3dframe}.

By construction the longest vertical edges are the edges connecting the top and bottom frames; their length is one unit more than the number of layers between the top and the bottom frame. Hence $\mu_v(D) = \chi(G) + 1$ by Lemma~\ref{lem:3dconst}.

The number of edges connecting the two layers is given by
\[
1+2+3+\cdots+(2n_G-2) = O(n_G^2)
\]
and is greater than $n_G^2$ for $n > 4$. The total edge length used by the L-gadgets is at most $2n_G^2$, and $\chi(G)m_G$ for the connecting segments. The top and bottom frames add $O(n_G^2)$ length to the drawing, and the connecting edges add $O(n_G^2)\chi(G)$. So we have that $n_G^2 \chi(G) \leq \tau(D) \leq O(n_G^2)\chi(G)$.
\end{proof}

\begin{figure}[b!]
\centering
\includegraphics[width=0.95\textwidth]{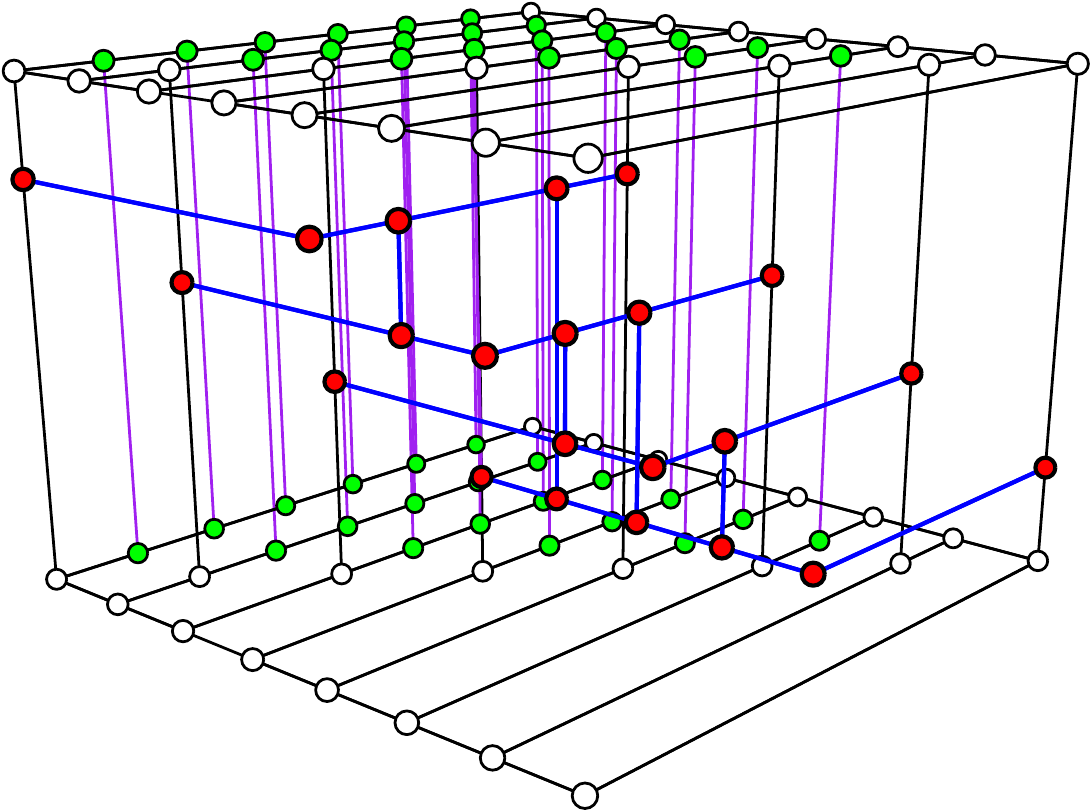}
\caption{3D frame construction
\label{fig:3dframe}}
\end{figure}
\newpage
\begin{theorem}
If $\P \neq \NP$, then there does not exist a polynomial time algorithm that
approximates the volume, longest vertical edge, or total edge length 
in an optimal layer compaction of a given three dimensional orthogonal drawing 
to within a factor of $n_D^{1/2-\epsilon}$.
\end{theorem}
\begin{proof}
Having the bounds in Lemma~\ref{lem:3dframe} the proof follows as before.
\end{proof}

We can also achieve a similar result for longest edge, without having to limit
our result to the longest \emph{vertical} edge, by replacing each L-shaped vertex gadget
by an L-bundle, similar to our vertex bundle in the two dimensional case. This is needed to prevent the longest edge from being a $n_G$ length leg of the L-gadget. We
omit the details.

\section{Coloring Counterexample}
In the conference version of this paper we included an approximation algorithm for compaction which was later discovered to be flawed. In this algorithm we defined a \emph{incompatibility graph}, where the vertices of the incompatibility graph represent sets of features in the drawing that must move in tandem. Two vertices of the incompatibility graph are connected by an edge when the drawing features they represent cannot be compacted to the same row of the drawing, i.e., when they contain parts of the drawing that are directly above one another. A coloring of the incompatibility graph was then used to to compact the drawing, using a row for each color. The problem is that although compacting a drawing does produce a coloring of its incompatibility graph, the converse does not hold. In Figure~\ref{fig:cexamp} we see a drawing whose incompatibility graph is 2-colorable (in fact it is a path), but which is completely uncompactable.

\begin{figure}[ht!]
\centering
\includegraphics[width=0.5\textwidth]{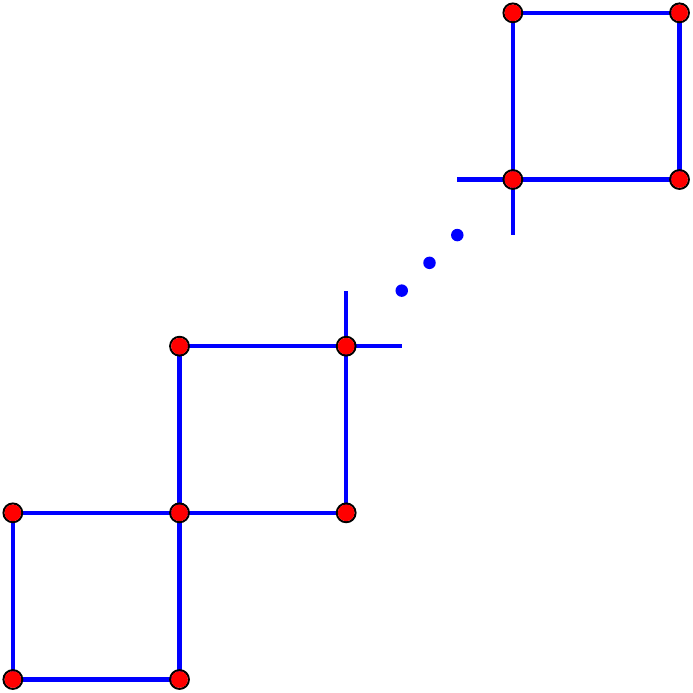}
\caption{An uncompactable graph with a 2-colorable incompatibility graph.}
\label{fig:cexamp}
\end{figure}

\section{Fixed-Parameter Tractability of \\Vertex-By-Vertex Vertical Compaction}

\begin{lemma}
\label{lem:fpt-fixed-num-of-layers}
Given an orthogonal drawing $D$ we can  compact $D$ into $k$ rows
in $O(k!(n_D+b_D))$ time, if such a compaction is possible.
\end{lemma}
\begin{proof}
We construct local assignments of the features into $k$ rows via a left-to-right plane sweep. The input drawing may be assumed to be given with coordinates on an $n\times n$ grid, so its features can be sorted in linear time using bucket sorting. While sweeping the plane we maintain the set of input features that intersect the sweep line along with a record of valid assignments of these features into the $k$ rows of our output drawing.

When a feature first intersects the sweep line we try to place it into the collection of valid assignments. If there are $\ell$ features intersecting the sweep line prior to the insertion, we have at most $\ell!\binom{k}{\ell}$ valid assignments to consider. In each of these valid assignments there are $k - \ell$ free rows. Altogether at most $k!$ configurations will be considered for each feature insertion. When the sweep line moves past a feature its row is freed for future use.

If at any point we cannot find any valid assignment for a new feature, we conclude that a compaction into $k$ rows is not possible. On the other hand if the last feature can be placed into a valid assignment, then a compaction into $k$ layers is possible. To recover the global assignment of horizontal features into rows, we may backtrack through the sets of  local assignments as is standard for a dynamic programming algorithm of this type.
\end{proof}

\begin{theorem}
An optimal vertex-by-vertex vertical compaction of an orthogonal drawing $D$ can be found in $O(\lambda!(n_D+b_D))$ time where $\lambda = \lambda(D)$.
\end{theorem}

\begin{proof}
Apply Lemma~\ref{lem:fpt-fixed-num-of-layers} for $k=1,2,3,\dots$ until finding a value of $k$ for which a valid layering exists.
\end{proof}

\section{Vertex-By-Vertex Free Compaction for Paths}
In contrast with row-by-row compaction (Section~\ref{sec:row-by-row}), we have a polynomial time algorithm for
vertex-by-vertex free compaction of a path. In fact, in the following theorem,
we give a simple, linear time algorithm to achieve an optimal compaction. In this section by a \emph{vertical subpath} we mean a maximal subpath containing only vertical edges, similarly for a \emph{horizontal subpath}

\begin{figure}[ht]
\centering
\ifFull
\includegraphics[width=.6\textwidth]{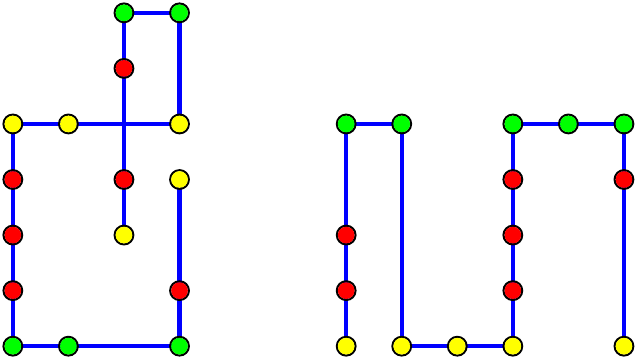}
\else
\includegraphics[width=.8\textwidth]{figures/path_twisted2}
\fi
\caption{Compacting a path into $k = 5$ rows. 
Here the colors on vertices are simply meant to help
clarify the correspondence between the two paths.}
\label{fig:path-untwist}
\end{figure}

\begin{theorem}
An optimal vertex-by-vertex free row compaction for an orthogonal path can by found in $O(n_D+b_D)$ time. Such a compaction will be into $k$ rows where $k$ is the number of vertices on the longest vertical path.
\end{theorem}
\newpage
\begin{proof}
First we replace the bends in the path with special ``bend vertices'' that will be changed back to bends after the compaction.
Now, the number of rows required by a compacted drawing must be at least
the number of vertices in the longest vertical subpath, since each of these
vertices must be placed in its own row to achieve an unambiguous drawing. We
will show that we can compact the drawing into exactly this number of rows, and therefore optimally compact the drawing.
To perform a vertex-by-vertex free compaction of an orthogonal path we perform the
following steps, which are illustrated in Figure~\ref{fig:path-untwist}.
\begin{enumerate}
\item
Compute the number $k$ of vertices is the longest vertical subpath by walking the path.
\item
Order the horizontal subpaths in walk order. Alternate placing horizontal paths, in path order, in rows $1$ and $k$ overlapping in their endpoints.
\item
Connect the overlapping horizontal paths with vertical lines, and place the appropriate   number of vertices on the edge to construct the vertical path between the overlapping horizontal paths.
\end{enumerate}
The end result is an optimally vertex-by-vertex free compacted path with respect to the number of rows.
\end{proof}

In this algorithm moving vertices horizontally was essential, so such a compaction is not possible in vertex-by-vertex vertical compaction. In contrast to our result, Brandes et al. showed that the problem of orthogonal order preserving drawing is $\NP$-hard for a path\cite{Brandes-Ortho}. We leave as an interesting open problem the complexity of vertex-by-vertex vertical compaction for a path.

\section{Conclusions}
Our investigations have determined lower bounds for several different approximation and fixed-parameter versions of the compaction problem. Currently we have no upper bounds on the approximation ratio for these compaction problems. With our techniques $O(\sqrt{n})$ is a natural upper bound, as our reductions produce drawings with at least one vertex for each edge in the original graph, approximately squaring the number of vertices in the problem. Can algorithms be constructed to match this expected upper bound? Can the lower bounds of $\Omega(n^{1/3-\epsilon})$ be improved to $\Omega(n^{1/2-\epsilon})$ or higher? We leave these questions open for future research.

\ifFull
\section*{Acknowledgments}
This work was supported in part by NSF grant 0830403 and by the Office of Naval Research under grant N00014-08-1-1015. A previous version of this paper appeared in the International Symposium on Graph Drawing 2011.
\fi

\raggedright
\bibliographystyle{abuser}
\bibliography{references}

\end{document}